\documentclass{article}%
\usepackage{amsmath}
\usepackage{amsfonts}
\usepackage{amssymb}
\usepackage{graphicx}
\usepackage{float}
\usepackage{tikz}
\usepackage{tkz-graph}%
\providecommand{\U}[1]{\protect\rule{.1in}{.1in}}
\newtheorem{theorem}{Theorem}[section]

\newtheorem{lemma}[theorem]{Lemma}

\newtheorem{remark}[theorem]{Remark}

\newenvironment{proof}[1][Proof]{\noindent \textbf{#1.} }{\  \rule{0.5em}{0.5em}}
\begin{document}

\title{Selfish Cops and Passive Robber:\\Qualitative Games}
\author{Ath. Kehagias and G. Konstantinidis }
\date{\today}
\maketitle

\begin{abstract}
Several variants of the \emph{cops and robbers} \ (CR)\ game have been studied
in the past. In this paper we examine a novel variant, which is played between
two cops, each one \emph{independently} trying to catch a \textquotedblleft
passive robber\textquotedblright. \ We will call this the \emph{Selfish Cops
and Passive Robber} \ (\emph{SCPR}) game. \noindent In short, SCPR\ is a
stochastic two-player, zero-sum game where \emph{the opponents are the two cop
players}. We study sequential and concurrent versions of the SCPR\ game. For
both cases we prove the existence of value and optimal strategies and present
algorithms for the computation of these.

\end{abstract}

\section{Introduction\label{sec01}}

Several variants of the \emph{cops and robbers} \ (CR)\ game have been studied
in the past. In this paper we examine a novel variant, which is played between
two cops, each one \emph{independently} trying to catch a \textquotedblleft
passive robber\textquotedblright. \ We will call this the \emph{Selfish Cops
and Passive Robber} \ (\emph{SCPR}) game. Here is a brief and informal
description of the game; a more detailed description will be provided in later sections.

\begin{enumerate}
\item The game is played on an undirected, finite, simple and connected graph.

\item The game is played by two cop players $C_{1}$ and $C_{2}$, each
controlling a \emph{cop token} (the tokens will also be referred to as $C_{1}$
and $C_{2}$).

\item A \emph{robber token} $R$ is also used, which is controlled by Chance.

\item At every turn of the game the tokens are moved from vertex to vertex,
along the edges of the graph.

\item The winner is the first player whose token lies at the same vertex as
the robber token (i.e., the player who \textquotedblleft captures the
robber\textquotedblright).
\end{enumerate}

\noindent In short, SCPR\ is a game where \emph{the opponents are the two cop
players}. As far as we know this CR\ variant has not been previously studied.

The study of Cops and Robbers was initiated by Quillot \cite{quilliot1978jeux}
and Nowakowski and Winkler \cite{nowakowski1983vertex}. For the graph
theoretic point of the view, the reader can consult the recent book
\cite{nowakowskibonato} which contains a good overview of the extensive
literature. We will study SCPR from a somewhat different angle, using the
theory of stochastic games\footnote{I.e., a sequence of normal-form games
where the game played at any time depends probabilistically on the previous
game played and the actions of the agents in that game.}  as presented in,
e.g., the book by Filar and Vrieze \cite{filar1997competitive}. In the current
paper we study \emph{qualitative} SCPR\ games, i.e., games in which the
\emph{payoff} is the \emph{winning probability}; in a forthcoming paper we
will discuss \emph{quantitative} SCPR\ games, i.e., games in which the payoff
is the \emph{expected capture time}.

The paper is organized as follows. In Section \ref{sec02} we present
preliminary definitions and notation. In Section \ref{sec03}\ we study the
\emph{sequential} version of the SCPR\ game and in Section \ref{sec04}\ the
\emph{concurrent} version. In Section \ref{sec05}\ we discuss connections to
several other research areas (for instance \emph{recursive games},
\emph{graphical games }and \emph{reachability games}) and in Section
\ref{sec06}\ we present concluding remarks and discuss future research directions.

\section{Preliminaries\label{sec02}}

The SCPR\ game is played on an undirected, finite, simple and connected graph
$G=\left(  V,E\right)  $, where $V$ is the vertex set and $E$ is the edge set.
Unless otherwise stated, we will assume that the \emph{cop number}%
\footnote{I.e., the minimum number of cops required to guarantee capture of
the robber.} of the graph is\ $c\left(  G\right)  =1$.

The game proceeds in turns numbered by $t\in\mathbb{N}_{0}$ and, as already
mentioned, involves three tokens: $C_{1},C_{2}$ and $R$. These will also be
referred to as the first, second and third token, respectively, and their
locations \emph{at the end} of the $t$-th turn are indicated by $X_{t}^{1}$,
$X_{t}^{2}$, $X_{t}^{3}$. \ The \ \emph{starting position} at the $0$-th turn
is \emph{given}: for $i\in\left\{  1,2,3\right\}  $, $X_{0}^{i}=x_{0}^{i}\in
V$. In subsequent turns, the positions are changed according to the rules of
the particular variant (sequential or concurrent) and subject to the
constraint that movement always follows the graph edges: $X_{t+1}^{i}\in
N\left[  X_{t}^{i}\right]  $ (the \emph{closed} neighborhood of $X_{t}^{i}$).
As will be seen in the sequel, in the general case token moves are governed by
\emph{probabilistic strategies}; hence $X_{t}^{1}$, $X_{t}^{2}$, $X_{t}^{3}$
are random variables.

\subsection{Sequential SCPR}

In the \emph{sequential} \emph{ }version of SCPR players take turns in moving
their tokens. More specifically, on odd-numbered turns $C_{1}$ is moved by the
first cop player; on even-numbered turns first $C_{2}$ is moved by the second
cop player and then $R$ is moved by Chance. Consequently, for $t=2l+1$ we have
$X_{t}^{2}=X_{t-1}^{2}$ and $X_{t}^{3}=X_{t-1}^{3}$; for $t=2l$ we have
$X_{t}^{1}=X_{t-1}^{1}$. An additional sequence of variables $U_{0}$, $U_{1}$,
$U_{2}$, $U_{3}$, $...$ indicates the player to move in the next turn; i.e.,
$U_{0}=U_{2}=...=1$, $U_{1}=U_{3}=...=2$. We also define the vector
$S_{t}=\left(  X_{t}^{1},X_{t}^{2},X_{t}^{3},U_{t}\right)  $.

A \emph{game position} or \emph{state} is a vector $s=\left(  x^{1}%
,x^{2},x^{3},u\right)  $ where $x^{1},x^{2},x^{3}\ $are the positions of the
three tokens and $u\ $indicates which cop is about to play. For instance,
$s=\left(  2,3,5,1\right)  $ denotes the situation in which $C_{1}$, $C_{2}$,
$R$ are located at vertices 2, 3, 5, respectively, and $C_{1}$ will move in
the next turn. We define the following sets of states%
\[
\forall i\in\left\{  1,2\right\}  :\mathbf{S}_{i}=V\times V\times
V\times\left\{  i\right\}  .
\]
In other words, $\mathbf{S}_{1}$ (resp. $\mathbf{S}_{2}$) is the set of states
\textquotedblleft belonging\textquotedblright\ to the first (resp.
second)$\ $player. A $C_{1}$-\emph{capture state} is an $s=\left(  x^{1}%
,x^{2},x^{3},u\right)  $ such that $x^{1}=x^{3}$. A $C_{2}$-\emph{capture
state}\footnote{$C_{1}$ is slightly favored, since an $\left(  x,x,x,u\right)
$ state \ is considered a $C_{1}$ capture; however, because of symmetry,
reversing the definitions of $C_{i}$-captures would yield essentially the same
results.} is an $s=\left(  x^{1},x^{2},x^{3},u\right)  $ such that
$x^{2}=x^{3}$ and $x^{1}\neq x^{3}$. We will also use a \emph{terminal }state,
denoted by $\tau$; the behavior of the terminal state will be described in
detail a little later. At any rate, the full \emph{state space} of the
sequential SCPR\ game is%
\[
\mathbf{S}=\mathbf{S}_{1}\cup\mathbf{S}_{2}\cup\left\{  \tau\right\}  .
\]

The random variable $A_{t}^{i}$ denotes the \emph{move} (or \emph{action})
\ of the $i$-th token at time $t$. When the game state is $s$, the set of
moves available to the $i$-th token is denoted by $\mathbf{A}_{i}\left(
s\right)  $. For instance, when $s=\left(  x^{1},x^{2},x^{3},1\right)  $ we
have $\mathbf{A}_{1}\left(  s\right)  =N\left[  x^{1}\right]  $ (the
\emph{closed} neigborhood of $x^{1}$) and $\mathbf{A}_{2}\left(  s\right)
=\left\{  x^{2}\right\}  $. Similar things hold for states $s=\left(
x^{1},x^{2},x^{3},2\right)  $. For $s=\tau$ we have $\mathbf{A}_{i}\left(
s\right)  =\left\{  \lambda\right\}  $, where $\lambda$ is the \emph{null
move}. Legal moves result to \textquotedblleft normal\textquotedblright\ state
transitions; e.g., suppose the current state is $s=\left(  2,3,5,1\right)  $
and the next moves are $a^{1}=3$, $a^{2}=3$, $a^{3}=5$; then, assuming $3\in
N\left[  2\right]  $, the next state is $s^{\prime}=\left(  3,3,5,2\right)  $.
However, the terminal state $\tau$ raises the following exceptions.

\begin{enumerate}
\item If the current state $s$ is a $C_{i}$-capture state ($i\in\left\{
1,2\right\}  $), then the next state is $s^{\prime}=\tau$, irrespective of the
token moves. In other words, a capture state always transits to the terminal state.

\item If the current state $s$ is the terminal (i.e., $s=\tau$), then the next
state is $s^{\prime}=\tau$ irrespective of the token moves. In other words,
the terminal always transits to itself.
\end{enumerate}

A \emph{play} or \emph{infinite history} of the SCPR\ game is an infinite
sequence $s_{0}s_{1}s_{2}...s_{n}...$ of game states. The set of all infinite
histories is denoted by
\[
H^{\infty}=\left\{  s_{0}s_{1}s_{2}...s_{t}...:s_{t}\in S\text{ for }%
t\in\mathbb{N}_{0}\right\}  .
\]
A \emph{finite history} is a sequence $s_{0}s_{1}s_{2}...s_{n}$ of game
states; the set of all histories of length $n$ is denoted by
\[
H_{n}=\left\{  s_{0}s_{1}s_{2}...s_{n-1}:s_{t}\in S\text{ for }t\in\left\{
0,1,...,n-1\right\}  \right\}  ;
\]
the set of all finite histories is $H=\cup_{n=0}^{\infty}H_{n}$.

We have already mentioned that each cop player moves his respective token.
Rather than specifying each move separately, we assume (as is usual in Game
Theory) that before the game starts, each cop player selects a \emph{strategy}
which controls all subsequent moves. Despite the fact that there is no robber
player, we will assume that robber movement is also controlled by a
\textquotedblleft strategy\textquotedblright, which has been fixed before the
game starts and \emph{is known to the cop players}. Hence the $i$-th token
($i\in\left\{  1,2,3\right\}  $) is controlled by the strategy (conditional
probability function):%
\[
\sigma_{i}\left(  a|s_{0}s_{1}...s_{t}\right)  =\Pr\left(  A_{t+1}%
^{i}=a|\left(  X_{0}^{1},X_{0}^{2},X_{0}^{3},U_{0}\right)  =s_{0},...,\left(
X_{t}^{1},X_{t}^{2},X_{t}^{3},U_{t}\right)  =s_{t}\right)  .
\]
The above definition is sufficiently general to describe every possible manner
of move selection. We will only consider strategies which assign zero
probability to illegal moves. The following classes of strategies are of
particular interest.

\begin{enumerate}
\item A strategy $\sigma_{i}$ is called \emph{stationary Markovian} (or
\emph{positional}) iff $\sigma_{i}\left(  a|s_{0}s_{1}...s_{t}\right)
=\sigma_{i}\left(  a|s_{t}\right)  $; i.e., the probability of the next move
depends only on the current \emph{state} of the game.

\item A strategy $\sigma_{i}$ is called \emph{oblivious} iff it is stationary
Markovian and $\sigma_{i}\left(  a|\left(  y^{1},y^{2},y^{3},u\right)
\right)  =\sigma_{i}\left(  a|y^{i},u\right)  $; i.e., the probability of the
next move of the token depends only on (i) the current \emph{location} of the
token and (ii)\ the active player.

\item A strategy $\sigma_{i}$ is called \emph{deterministic }iff, for every
$s_{0}s_{1}...s_{t}\in H$, $\sigma_{i}\left(  x|s_{0}s_{1}...s_{t}\right)
\in\left\{  0,1\right\}  $; i.e., for every history, the $i$-th token moves to
its next location deterministically.
\end{enumerate}

To simplify presentation, we will often use the following notation for
deterministic strategies. We define the \emph{deterministic strategy }to be a
function $\overline{\sigma}_{i}:H\rightarrow V$, defined  as follows: for
every finite history $s_{0}s_{1}...s_{t}$, $\overline{\sigma}_{i}\left(
s_{0}s_{1}...s_{t}\right)  =a$, where $a$ is the unique vertex such that
$\sigma_{i}\left(  a|s_{0}s_{1}...s_{t}\right)  =1$. If $\sigma_{i}$ is
stationary Markovian, we write $\overline{\sigma}_{i}\left(  s_{t}\right)  =a$.

Suppose the game is in state $s$. Now $C_{1}$ plays $a^{1}$, $C_{2}$ plays
$a^{2}$ and $R$'s move $a^{3}$ is selected according to the (fixed)\ strategy
$\sigma_{3}$; hence the game will move into some new state $s^{\prime}\ $with
a certain probability depending on $a_{1},a_{2}$ and $\sigma_{3}$. We denote
this probability by $\Pr\left(  s^{\prime}|s,a_{1},a_{2}\right)
$\footnote{Note that in all subsequent notation, the dependence on the fixed
and known $\sigma_{3}$ is suppressed. Also, when a cop reaches the vertex
occupied by the robber we have a capture with probability one, irrespective of
the robber's move. }.

Payoff is defined as follows. In each turn of the game, $C_{1}$ receives an
\emph{immediate payoff} equal to
\begin{equation}
q\left(  s\right)  =\left\{
\begin{array}
[c]{rl}%
1 & \text{iff }s\text{ is a }C_{1}\text{-capture state}\\
0 & \text{otherwise}%
\end{array}
\right.  ; \label{eq03001}%
\end{equation}
$C_{2}$ receives $-q\left(  s\right)  $. Hence, a play $s_{0}s_{1}....$
results in (total) \emph{payoff}%
\begin{equation}
Q\left(  s_{0}s_{1}....\right)  =\sum_{t=0}^{\infty}q\left(  s_{t}\right)
\label{eq03003}%
\end{equation}
for $C_{1}$ and $-Q\left(  s_{0}s_{1}....\right)  $ for $C_{2}$. Note that
both players have an incentive to capture $R$.

\begin{enumerate}
\item If $C_{1}$ captures the robber, he receives a total payoff of one
(comprising of immediate payoff of one for the capture turn and zero for all
other turns); otherwise his total payoff is zero.

\item $C_{2}$ \emph{never} receives positive payoff (even if he captures the
robber). However, we have assumed $c\left(  G\right)  =1$ and this implies
that a single cop can always catch the robber. Hence, if $C_{2}$ does not
capture $R$, $C_{1}$ will and thus $C_{2}$ will receive a negative payoff;
this provides the incentive for $C_{2}$ to capture $R$.
\end{enumerate}

Sequential SCPR\ is a \emph{stochastic zero sum game}
\cite{filar1997competitive}. Each player will try to maximize his
\emph{expected} payoff. Suppose the game starts at position $s_{0}$, $C_{i}$
moves according to strategy \ $\sigma_{i}$ (for $i\in\left\{  1,2\right\}  $)
and $R$ moves according to a fixed and known stratey $\sigma_{3}$. Every
triple $\left(  \sigma_{1},\sigma_{2},\sigma_{3}\right)  $ induces a
probability measure on $H^{\infty}$, the set of all infinite game histories.
Hence the \emph{expected} payoff to $C_{1}$ is%
\begin{equation}
J\left(  \sigma_{1},\sigma_{2}|s_{0}\right)  =\mathbb{E}\left(  \sum
_{t=0}^{\infty}q\left(  s_{t}\right)  |\left(  X_{0}^{1},X_{0}^{2},X_{0}%
^{3},U_{0}\right)  =s_{0}\right)  \label{eq03002}%
\end{equation}
and is well defined; $-J\left(  \sigma_{1},\sigma_{2}|s_{0}\right)  $ is the
expected payoff to $C_{2}$. It is easily seen that
\[
J\left(  \sigma_{1},\sigma_{2}|s_{0}\right)  =\Pr\left(
\text{\textquotedblleft}C_{1}\text{ wins\textquotedblright}%
|\text{\textquotedblleft the game starts at }s_{0}\text{ and, for }%
i\in\left\{  1,2\right\}  \text{, }C_{i}\text{ uses }\sigma_{i}%
\text{\textquotedblright}\right)  .
\]
We always have
\begin{equation}
\sup_{\sigma_{1}}\inf_{\sigma_{2}}J\left(  \sigma_{1},\sigma_{2}|s_{0}\right)
\leq\inf_{\sigma_{2}}\sup_{\sigma_{1}}J\left(  \sigma_{1},\sigma_{2}%
|s_{0}\right)  ; \label{eq03003a}%
\end{equation}
if the two sides of (\ref{eq03003a}) are equal, we define the \emph{value} of
the game (when started at $s_{0}$)\ to be%
\begin{equation}
v\left(  s_{0}\right)  =\sup_{\sigma_{1}}\inf_{\sigma_{2}}J\left(  \sigma
_{1},\sigma_{2}|s_{0}\right)  =\inf_{\sigma_{2}}\sup_{\sigma_{1}}J\left(
\sigma_{1},\sigma_{2}|s_{0}\right)  ; \label{eq03003b}%
\end{equation}
$\mathbf{v}$ will denote the vector of values for all starting states, i.e.,
$\mathbf{v=}\left(  v\left(  s\right)  \right)  _{s\in\mathbf{S}}$. Given some
$\varepsilon\geq0$, we say that:

\begin{enumerate}
\item a strategy $\sigma_{1}^{\#}$ is $\varepsilon$\emph{-optimal} (for
$C_{1}$) iff $\forall s_{0}:v\left(  s_{0}\right)  -\inf_{\sigma_{2}}J\left(
\sigma_{1}^{\#},\sigma_{2}|s_{0}\right)  \leq\varepsilon$;

\item a strategy $\sigma_{2}^{\#}$ is $\varepsilon$\emph{-optimal} (for
$C_{2}$) iff $\forall s_{0}:v\left(  s_{0}\right)  -\sup_{\sigma_{1}}J\left(
\sigma_{1},\sigma_{2}^{\#}|s_{0}\right)  \geq-\varepsilon.$
\end{enumerate}

\noindent A 0-optimal strategy is also called simply \emph{optimal}.

Finally, let $\Gamma$ be a \emph{matrix game}, i.e., a (one-turn) two-player,
zero-sum game with finite action set $\mathbf{A}_{i}$ for the $i$-th player
and the payoff to the first player being $\Gamma\left(  a^{1},a^{2}\right)  $
when $i$-th player plays $a^{i}\in\mathbf{A}_{i}$ (with $i\in\left\{
1,2\right\}  $). As is well known \cite{osborne1994game}, such a game always
has a \emph{value}, which we will denote by $\mathbf{Val}\left[  \Gamma\left(
a^{1},a^{2}\right)  \right]  $.

\subsection{Concurrent SCPR}

Most of the CR\ literature studies sequential versions of the CR\ game.
However, we have recently introduced a \emph{concurrent} version of the
classic CR\ \cite{kehagias2016}. Now we extend concurrency to the SCPR\ game.

The concurrent SCPR\ game differs from the sequential game in a basic
respect:\ in every turn the $C_{1}$, $C_{2}$, $R$ tokens are moved
\emph{simultaneously} (hence, when making his move, each player does not know
the other player's move; note that both of them know $R$'s next move, since
$\sigma_{3}$ is known in advance). Once again we will assume, unless otherwise
indicated, that $\widehat{c}\left(  G\right)  =1$ (note that a graph $G$ has
\emph{concurrent cop number} $\widehat{c}\left(  G\right)  =k$ iff it has
sequential, i.e., \textquotedblleft classic\textquotedblright, $\ $cop number
$c\left(  G\right)  =k$ \cite{kehagias2016}).

In addition, in concurrent SCPR we can have \textquotedblleft\emph{en-passant
capture}\textquotedblright, in which a cop and the robber start at opposite
ends of the same edge and move in opposite directions; in this case the robber
is \textquotedblleft swept\textquotedblright\ by the cop and moved into the
cop's destination.

With concurrent movement, game states are vectors $\left(  x^{1},x^{2}%
,x^{3}\right)  $ where $x^{i}\in V$ indicates (as previously) the position of
the $i$-th token; the $u$ variable is no longer necessary, since all tokens
are moved in every turn. Capture states now have the form $\left(  x^{1}%
,x^{2},x^{3}\right)  $ with either $x^{1}=x^{3}$ or $x^{2}=x^{3}$ (or both)
and the definition and behavior of the terminal state $\tau$ are the same as
previously. For the state space, we define
\[
\widehat{\mathbf{S}}_{a}=V\times V\times V,\quad\widehat{\mathbf{S}}%
=\widehat{\mathbf{S}}_{a}\cup\left\{  \tau\right\}
\]
and $\widehat{\mathbf{S}}$ is the full \emph{state space} of the of concurrent
SCPR game.

Regarding $\mathbf{A}_{i}\left(  s\right)  $ (the actions available to the
$i$-th player when the game state is $s$) we always have $\mathbf{A}%
_{i}\left(  \left(  x^{1},x^{2},x^{3}\right)  \right)  \in N\left[
x^{i}\right]  $. The definitions of (finite and infinite) histories and
strategies are the same as in the sequential case, except that we now use the
state space $\widehat{\mathbf{S}}$. The meaning of the sets $\widehat{H}_{n}$,
$\widehat{H}$, $\widehat{H}^{\infty}$ is analogous to that of $H_{n}$, $H$,
$H^{\infty}$. The strategies $\sigma_{i}$ ($i\in\left\{  1,2,3\right\}  $) are
defined in the same manner as in the sequential case (again, for deterministic
moves we introduce the deterministic strategy functions $\overline{\sigma}%
_{i}$).

Payoff of the concurrent SCPR\ game is defined in exactly the same manner as
in the sequential case. Again, concurrent SCPR\ is a stochastic zero sum game
and each player will try to maximize his expected payoff.

\section{Sequential SCPR\label{sec03}}

In this section we establish that sequential SCPR\ has a \emph{value} which
can be computed by \emph{value iteration}.

\begin{theorem}
\label{prop0301}Given some graph $G=\left(  V,E\right)  $. For every
$s\in\mathbf{S}_{1}\cup\mathbf{S}_{2}$, the sequential SCPR\ game starting at
$s$ has a value $v\left(  s\right)  $. The vector of values $\mathbf{v=}%
\left(  v\left(  s\right)  \right)  _{s\in\mathbf{S}}$ is the smallest
(componentwise) solution of the following \emph{optimality equations}:%
\begin{align}
&  v\left(  \tau\right)  =0;\quad\label{eq03006a}\\
\forall s=\left(  x^{1},x^{2},x^{3},1\right)  \in\mathbf{S}_{1}:\quad &
v\left(  s\right)  =\max_{a^{1}}\left[  q\left(  s\right)  +\sum_{s^{\prime
}\in\mathbf{S}}\Pr\left(  s^{\prime}|s,a^{1},x^{2}\right)  v\left(  s^{\prime
}\right)  \right]  ,\label{eq03006}\\
\forall s=\left(  x^{1},x^{2},x^{3},2\right)  \in\mathbf{S}_{2}:\quad &
v\left(  s\right)  =\min_{a^{2}}\left[  q\left(  s\right)  +\sum_{s^{\prime
}\in\mathbf{S}}\Pr\left(  s^{\prime}|s,x^{1},a^{2}\right)  v\left(  s^{\prime
}\right)  \right]  . \label{eq03007}%
\end{align}
Furthermore $C_{2}$ has a deterministic stationary Markovian optimal strategy
and, for every $\varepsilon>0$, $C_{1}$ has a deterministic stationary
Markovian $\varepsilon$-optimal strategy.
\end{theorem}

\begin{proof}
It is easily checked that, for every graph $G$ and every starting position
$s$, the sequential SCPR game is a \emph{positive} zero sum stochastic game.
Hence (by \cite[Theorem 4.4.1]{filar1997competitive}) it possesses a value
which (by \cite[Theorem 4.4.3]{filar1997competitive}) is the smallest
componentwise\ solution to the following system of \emph{optimality
equations}:%
\begin{equation}
v\left(  \tau\right)  =0;\quad\forall s\in\mathbf{S}_{1}\cup\mathbf{S}%
_{2}:v\left(  s\right)  =\mathbf{Val}\left[  q\left(  s\right)  +\sum
_{s^{\prime}\in S}\Pr\left(  s^{\prime}|s,a^{1},a^{2}\right)  v\left(
s^{\prime}\right)  \right]  .\label{eq0321}%
\end{equation}
However, in each turn of the sequential SCPR game, one of the players has a
single available action. For instance, when the state is $s=\left(
x^{1},x^{2},x^{3},1\right)  $, $C_{2}$'s action set can only be $a^{2}=x^{2}$.
Hence in (\ref{eq0321}) we are taking the value of an one-shot game with \ the
game matrix consisting of a single column. It follows that
\[
\forall s=\left(  x^{1},x^{2},x^{3},1\right)  :\mathbf{Val}\left[  q\left(
s\right)  +\sum_{s^{\prime}\in\mathbf{S}}\Pr\left(  s^{\prime}|s,a^{1}%
,a^{2}\right)  v\left(  s^{\prime}\right)  \right]  =\max_{a^{1}}\left[
q\left(  s\right)  +\sum_{s^{\prime}\in\mathbf{S}}\Pr\left(  s^{\prime
}|s,a^{1},x^{2}\right)  v\left(  s^{\prime}\right)  \right]
\]
which proves (\ref{eq03006}); (\ref{eq03007}) can be proved similarly.

The existence of stationary Markovian optimal strategy for $C_{2}$ follows
from \cite[Corollary 4.4.2]{filar1997competitive}. It is a deterministic
strategy because for each state $s\in\mathbf{S}_{2}$ the corresponding optimal
$C_{2}$ move is the one minimizing (\ref{eq03007}). Similarly, the existence
of a stationary Markovian $\varepsilon$-optimal strategy for $C_{1}$ follows
from \cite[Problem 4.16]{filar1997competitive}; the strategy is deterministic,
because for each state $s\in\mathbf{S}_{1}$ the corresponding optimal $C_{1}$
move is the one maximizing (\ref{eq03006}).
\end{proof}

For the computation of the solution to (\ref{eq03006})-(\ref{eq03007}) we have
the following.

\begin{theorem}
\label{prop0302}Given some graph $G=\left(  V,E\right)  $. Define
$\mathbf{v}^{\left(  0\right)  }$ by
\[
v^{\left(  0\right)  }\left(  \tau\right)  =0;\quad\forall s\in\mathbf{S}%
_{1}\cup\mathbf{S}_{2}:v^{\left(  0\right)  }\left(  s\right)  =q\left(
s\right)
\]
and $\mathbf{v}^{\left(  1\right)  },\mathbf{v}^{\left(  2\right)  },...$ by
the following recursion
\begin{align}
&  v^{\left(  i\right)  }\left(  \tau\right)  =0\label{eq03008}\\
\forall s=\left(  x^{1},x^{2},x^{3},1\right)  \in\mathbf{S}_{1}:  &  \quad
v^{\left(  i\right)  }\left(  s\right)  =\max_{a^{1}}\left[  q\left(
s\right)  +\sum_{s^{\prime}\in\mathbf{S}}\Pr\left(  s^{\prime}|s,a^{1}%
,x^{2}\right)  v^{\left(  i-1\right)  }\left(  s^{\prime}\right)  \right]  ,\\
\forall s=\left(  x^{1},x^{2},x^{3},2\right)  \in\mathbf{S}_{2}:  &  \quad
v^{\left(  i\right)  }\left(  s\right)  =\min_{a^{2}}\left[  q\left(
s\right)  +\sum_{s^{\prime}\in\mathbf{S}}\Pr\left(  s^{\prime}|s,x^{1}%
,a^{2}\right)  v^{\left(  i-1\right)  }\left(  s^{\prime}\right)  \right]  .
\label{eq03009}%
\end{align}
Then, for every $s\in\mathbf{S}_{1}\cup\mathbf{S}_{2}$, $\lim_{i\rightarrow
\infty}v^{\left(  i\right)  }\left(  s\right)  $ exists and equals $v\left(
s\right)  $, the value of the sequential SCPR\ game played on $G$, starting
from $s$.
\end{theorem}

\begin{proof}
Obviously, for all $s\in\mathbf{S}$, $v\left(  s\right)  \in\left[
0,1\right]  $. Hence $\mathbf{v}$ is a (componentwise)\ finite vector. Then
from \cite[Theorem 4.4.4]{filar1997competitive} we know that, defining
$\mathbf{v}^{\left(  0\right)  }$ by
\[
v^{\left(  0\right)  }\left(  \tau\right)  =0;\quad\forall s\in\mathbf{S}%
_{1}\cup\mathbf{S}_{2}:v^{\left(  0\right)  }\left(  s\right)  =q\left(
s\right)
\]
and $\mathbf{v}^{\left(  1\right)  },\mathbf{v}^{\left(  2\right)  },...$ by
the recursion
\begin{equation}
v^{\left(  i\right)  }\left(  \tau\right)  =0;\quad\forall s\in\mathbf{S}%
_{1}\cup\mathbf{S}_{2}:\quad v^{\left(  i\right)  }\left(  s\right)
=\mathbf{Val}\left[  q\left(  s\right)  +\sum_{s^{\prime}\in S}\Pr\left(
s^{\prime}|s,a_{1},a_{2}\right)  v^{\left(  i-1\right)  }\left(  s^{\prime
}\right)  \right]  , \label{eq03010}%
\end{equation}
we get $\lim_{i\rightarrow\infty}\mathbf{v}^{\left(  i\right)  }=\mathbf{v}$
(the value vector of Theorem \ref{prop0301}). The equivalence of
(\ref{eq03010}) to (\ref{eq03008})-(\ref{eq03009}) is established by the
argument used in the proof of Theorem \ref{prop0301}.
\end{proof}

\begin{remark}
\normalfont The significance of Theorem \ref{prop0301} is the following. Since
SCPR\ is a positive zero sum stochastic game, it will certainly have a value,
which satisfies the optimality equations (\ref{eq0321}); each equation of the
system (\ref{eq0321}) involves the value of a one-turn game. However, the
optimality equations can be expressed in the simpler form (\ref{eq03006}%
)-(\ref{eq03007}) which show that the values of the one-step games can be
computed by simple max and min operations.
\end{remark}

\begin{remark}
\normalfont Similar remarks can be made about Theorem \ref{prop0302}, where
the iteration (\ref{eq03008})-(\ref{eq03009}) is computationally simpler
(involves only max and min operations) than (\ref{eq03010}). Note the
similarity of (\ref{eq03008})-(\ref{eq03009}) to the algorithm of
\cite{hahn2006game} for determining the winner of a classic CR\ game. The
similarity becomes stronger in the case of deterministic $\sigma_{3}$. In this
case, $\Pr\left(  s^{\prime}|s,a^{1},x^{2}\right)  $ equals 1 for a single
$s^{\prime}=\mathbf{T}\left(  s,a^{1},x^{2}\right)  $ and $\Pr\left(
s^{\prime}|s,x^{1},a^{2}\right)  $ equals 1 for a single $s^{\prime
}=\mathbf{T}\left(  s,x^{1},a^{2}\right)  $; where $\mathbf{T}\left(
s,a^{1},a^{2}\right)  $ is the \emph{transition function} which yields the
next state when, from $s$, $C_{1}$ plays $a^{1}$ and $C_{2}$ plays $a^{2}$;
there is also a suppressed dependence on the move of $R$, which is
$\overline{\sigma}_{3}\left(  s\right)  $. Using this notation, (\ref{eq03008}%
)-(\ref{eq03009}) simplify to
\begin{align}
\forall s  &  =\left(  x^{1},x^{2},x^{3},1\right)  \in\mathbf{S}%
_{1}:v^{\left(  i\right)  }\left(  s\right)  =\max_{a^{1}}\left[  q\left(
s\right)  +v^{\left(  i-1\right)  }\left(  \mathbf{T}\left(  s,a^{1}%
,x^{2}\right)  \right)  \right]  ,\label{eq03008a}\\
\forall s  &  =\left(  x^{1},x^{2},x^{3},2\right)  \in\mathbf{S}%
_{2}:v^{\left(  i\right)  }\left(  s\right)  =\min_{a^{2}}\left[  q\left(
s\right)  +v^{\left(  i-1\right)  }\left(  \mathbf{T}\left(  s,x^{1}%
,a^{2}\right)  \right)  \right]  ; \label{eq03009a}%
\end{align}
these parallel closely the algorithm of \cite[p.2494]{hahn2006game}.
\end{remark}

\begin{remark}
\normalfont Finally, note that Theorems \ref{prop0301} and \ref{prop0302} hold
even when $c\left(  G\right)  >1$; the reason for which we have previously
required $c\left(  G\right)  =1$ has to do with the appropriateness of the
payoff function introduced in Section \ref{sec02}. In particular, when
$c\left(  G\right)  >1$ our argument about $C_{2}$'s incentive to capture $R$
does not hold necessarily (i.e., depending on $\sigma_{3}$, $C_{2}$ may ensure
payoff of 0 without ever capturing $R$); Theorems \ref{prop0301} and
\ref{prop0302} still hold true.
\end{remark}

\section{Concurrent SCPR\label{sec04}}

In this section we establish that concurrent SCPR\ has a value which can be
computed by value iteration. We first consider the case in which $R$ is
controlled by a general probability function $\sigma_{3}$ (\textquotedblleft
random robber\textquotedblright) and then examine in greater detail the case
in which $\sigma_{3}$ is oblivious deterministic (\textquotedblleft oblivious
deterministic robber\textquotedblright).

\subsection{Random Robber\label{sec0301}}

The two main results on concurrent SCPR are immediate consequences of the more
general results of \cite{filar1997competitive}.

\begin{theorem}
\label{prop04001}Given some graph $G=\left(  V,E\right)  $. For every
$s=\left(  x^{1},x^{2},x^{3}\right)  \in\widehat{\mathbf{S}}_{a}$, the
concurrent SCPR\ game starting at $s$ has a value $v\left(  s\right)  $. The
vector of values $\mathbf{v=}\left(  v\left(  s\right)  \right)
_{s\in\mathbf{S}}$ is the smallest (componentwise)\ solution of the following
\emph{optimality }equations%
\begin{equation}
v\left(  \tau\right)  =0;\quad\forall s\in\widehat{\mathbf{S}}_{a}:v\left(
s\right)  =\mathbf{Val}\left[  q\left(  s\right)  +\sum_{s^{\prime}\in S}%
\Pr\left(  s^{\prime}|s,a_{1},a_{2}\right)  v\left(  s^{\prime}\right)
\right]  .\label{eq03012}%
\end{equation}
Furthermore, $C_{2}$ has a stationary Markovian optimal strategy and, for
every $\varepsilon>0$, $C_{1}$ has a stationary Markovian $\varepsilon
$-optimal strategy.
\end{theorem}

\begin{proof}
For every graph $G$ (and every starting position $s$) SCPR is a positive
stochastic game. Hence (by \cite[Theorem 4.4.1]{filar1997competitive}) it
possesses a value which (by \cite[Theorem 4.4.3]{filar1997competitive})
satisfies the optimality equation (\ref{eq03012}). Furthermore $C_{2}$ has a
stationary Markovian optimal strategy by \cite[Corollary 4.4.2]%
{filar1997competitive} and, for every $\varepsilon>0$, $C_{1}$ has a
stationary Markovian $\varepsilon$-optimal strategy by \cite[Problem
4.16]{filar1997competitive}.
\end{proof}

\begin{theorem}
\label{prop04002}Given some graph $G=\left(  V,E\right)  $, let $s=\left(
x^{1},x^{2},x^{3}\right)  \in\widehat{\mathbf{S}}_{a}$. Define $\mathbf{v}%
^{\left(  0\right)  }$ by
\[
v^{\left(  0\right)  }\left(  \tau\right)  =0;\quad\forall s\in\widehat
{\mathbf{S}}_{a}:v^{\left(  0\right)  }\left(  s\right)  =q\left(  s\right)
\]
and $\mathbf{v}^{\left(  1\right)  },\mathbf{v}^{\left(  2\right)  },...$ by
the following recursion
\begin{equation}
v^{\left(  i\right)  }\left(  \tau\right)  =0;\quad\forall s\in\widehat
{\mathbf{S}}_{a}:v^{\left(  i\right)  }\left(  s\right)  =\mathbf{Val}\left[
q\left(  s\right)  +\sum_{s^{\prime}\in\mathbf{S}}\Pr\left(  s^{\prime
}|s,a_{1},a_{2}\right)  v^{\left(  i-1\right)  }\left(  s^{\prime}\right)
\right]  . \label{eq03011}%
\end{equation}
Then, for every $s\in\widehat{\mathbf{S}}_{a}$, $\lim_{i\rightarrow\infty
}v^{\left(  i\right)  }\left(  s\right)  $ exists and equals $v\left(
s\right)  $, the value of the concurrent SCPR\ game played on $G$, starting
from $s$.
\end{theorem}

\begin{proof}
This follows immediately from \cite[Theorem 4.4.4]{filar1997competitive}.
\end{proof}

\subsection{Oblivious Deterministic Robber\label{sec0302}}

Theorems \ref{prop0301} and \ref{prop0302} are \textquotedblleft
simpler\textquotedblright\ than Theorems \ref{prop04001} and \ref{prop04002},
in the sense that the former do not involve the computation of matrix game
values. We will now show that, when $\sigma_{3}$ is oblivious deterministic,
we can obtain a similar simplification of Theorems \ref{prop04001} and
\ref{prop04002}. Before presenting these results in rigorous form, let us
describe them informally.

\begin{enumerate}
\item Suppose first that a game is played between a single cop and an
oblivious deterministic robber. We will prove that there exists a stationary
Markovian deterministic cop strategy $\overline{\sigma}^{\ast}$ by which the
cop can capture the robber in minimum time.

\item Next consider two cops and an oblivious deterministic robber. We will
prove that the extension of $\overline{\sigma}^{\ast}$ to SCPR is
\emph{optimal} for both cops. More specifically, neither cop loses anything by
using it; and one of the two will capture the robber with probability one.
\end{enumerate}

Let us now formalize the above ideas. We pick any graph $G=\left(  V,E\right)
$ and any oblivious deterministic robber strategy $\overline{\sigma}_{3}$ and
keep these fixed for the remainder of the discussion. Further, let
$\mathcal{S}$ denote the set of all functions $\overline{\sigma}:V\times
V\rightarrow V$ with the restriction that $\forall\left(  x^{1},x^{3}\right)
\in V\times V:\overline{\sigma}\left(  x^{1},x^{3}\right)  \in N\left[
x^{1}\right]  $. In other words, $\mathcal{S}$ is the set of legal stationary
Markovian deterministic cop strategies for the \textquotedblleft
classic\textquotedblright\ CR game of one cop and one robber.

Now pick some $\overline{\sigma}\in\mathcal{S}$ and play the game with
starting positions $X_{0}^{1}=x_{0}^{1}\in V$ (for the cop) and $X_{0}%
^{3}=x_{0}^{3}\in V$ (for the robber). The following sequence (dependent on
$\overline{\sigma}$, $x_{0}^{1}$, $x_{0}^{3}$) of cop and robber positions
will be produced:
\[
X_{0}^{1}=x_{0}^{1},X_{0}^{3}=x_{0}^{3},X_{1}^{1}=\overline{\sigma}\left(
x_{0}^{1},x_{0}^{3}\right)  ,X_{1}^{3}=\overline{\sigma}_{3}\left(  x_{0}%
^{3}\right)  ,...\text{ ;}%
\]
let $T_{\overline{\sigma}}\left(  x_{0}^{1},x_{0}^{3}\right)  $ be the
\emph{capture time}, i.e., the smallest $t$ such that $X_{t}^{1}=X_{t}^{3}$,
for the sequence produced by $\overline{\sigma}$, $x_{0}^{1}$, $x_{0}^{3}$
(and $\overline{\sigma}_{3}$). Also define
\[
\overline{V\times V}=\left\{  \left(  x^{1},x^{3}\right)  :x^{1}\in V,x^{3}\in
V,x^{1}\neq x^{3}\right\}  .
\]
Then we have the following.

\begin{lemma}
\label{prop0303}Given a graph $G=\left(  V,E\right)  $ and an oblivious
deterministic robber strategy $\overline{\sigma}_{3}$. Let
\[
\forall x^{1}\in V:T^{\left(  0\right)  }\left(  x^{1},x^{1}\right)
=0,\quad\forall\left(  x^{1},x^{3}\right)  \in\overline{V\times V}:T^{\left(
0\right)  }\left(  x^{1},x^{3}\right)  =\infty.
\]
Now perform the following iteration for $i=1,2,...$ :%
\begin{align}
\forall x^{1} &  \in V:T^{\left(  i\right)  }\left(  x^{1},x^{1}\right)
=0;\quad\forall\left(  x^{1},x^{3}\right)  \in\overline{V\times V}:T^{\left(
i\right)  }\left(  x^{1},x^{3}\right)  =\min_{x^{\prime}\in N\left[
x^{1}\right]  }\left[  1+T^{\left(  i-1\right)  }\left(  x^{\prime},\sigma
_{3}\left(  x^{3}\right)  \right)  \right]  ,\label{eq03024}\\
\forall x^{1} &  \in V:T^{\left(  i\right)  }\left(  x^{1},x^{1}\right)
=0;\quad\forall\left(  x^{1},x^{3}\right)  \in\overline{V\times V}%
:\overline{\sigma}^{\left(  i\right)  }\left(  x^{1},x^{3}\right)  =\arg
\min_{x^{\prime}\in N\left[  x^{1}\right]  }\left[  1+T^{\left(  i-1\right)
}\left(  x^{\prime},\sigma_{3}\left(  x^{3}\right)  \right)  \right]
.\label{eq03025}%
\end{align}
Then the limits
\[
\lim_{i\rightarrow\infty}\overline{\sigma}^{\left(  i\right)  }\left(
x^{1},x^{3}\right)  ,\qquad\lim_{i\rightarrow\infty}T^{\left(  i\right)
}\left(  x^{1},x^{3}\right)
\]
exist for all$\ \left(  x^{1},x^{3}\right)  \in V\times V$. Furthermore,
letting $\overline{\sigma}^{\ast}\left(  x^{1},x^{3}\right)  =\lim
_{i\rightarrow\infty}\overline{\sigma}^{\left(  i\right)  }\left(  x^{1}%
,x^{3}\right)  $ and $T^{\ast}\left(  x^{1},x^{3}\right)  =\min_{\overline
{\sigma}\in\mathcal{S}}T_{\overline{\sigma}}\left(  x^{1},x^{3}\right)  $, we
have%
\begin{equation}
\forall\left(  x^{1},x^{3}\right)  \in V\times V:\lim_{i\rightarrow\infty
}T^{\left(  i\right)  }\left(  x^{1},x^{3}\right)  =T_{\overline{\sigma}%
^{\ast}}\left(  x^{1},x^{3}\right)  =T^{\ast}\left(  x^{1},x^{3}\right)
.\label{eq03022}%
\end{equation}

\end{lemma}

\begin{proof}
The proof is based on a standard dynamic programming argument. First note
that, for every $\left(  x^{1},x^{3}\right)  \in V\times V$, $T^{\ast}\left(
x^{1},x^{3}\right)  <\left\vert V\right\vert $. This is true because $C_{1}$
can reach any vertex of $V$ in at most $\left\vert V\right\vert -1$ moves; so
$C_{1}$ can simply go to $X_{\left\vert V\right\vert }^{3}$ (the \emph{known}
location of $R$ at time $t=\left\vert V\right\vert $) and wait for the robber there.

Next we prove by induction that
\begin{equation}
T^{\ast}\left(  x^{1},x^{3}\right)  =n\Rightarrow\left(  \forall i\geq
n:T^{\ast}\left(  x^{1},x^{3}\right)  =T^{\left(  i\right)  }\left(
x^{1},x^{3}\right)  \right)  .\label{eq03023}%
\end{equation}
For $n=0$, $T^{\ast}\left(  x^{1},x^{3}\right)  =0$ implies $x^{1}=x^{3}$ and,
from the algorithm, $T^{\ast}\left(  x^{1},x^{1}\right)  =0=T^{\left(
i\right)  }\left(  x^{1},x^{1}\right)  $ for all $i\in\mathbb{N}_{0}$. Now
suppose that (\ref{eq03023}) holds for $n=1,2,...,k$ and consider the case
$n=k+1$, in which $T^{\ast}\left(  x^{1},x^{3}\right)  =k+1$ is the smallest
number of steps in which $C_{1}$ can reach $R$. This also means that
(i)\ there exists some $x^{\prime}\in N\left[  x^{1}\right]  $ from which
$C_{1}$ can reach $R$ (who now starts at $\overline{\sigma}_{3}\left(
x^{3}\right)  $) in $k$ steps and (ii) there does not exist any $x^{\prime
\prime}\in N\left[  x^{1}\right]  $ from which $C_{1}$ can reach $R$ in $m<k$
steps (because then $C_{1}$ starting at $x^{1}$ could reach $R$ in $m+1<k+1$
steps). In other words%
\[
T^{\ast}\left(  x^{1},x^{3}\right)  =k+1\Rightarrow T^{\ast}\left(
x^{1},x^{3}\right)  =\min_{x^{\prime}\in N\left[  x^{1}\right]  }\left[
1+T^{\left(  k\right)  }\left(  x^{\prime},\overline{\sigma}_{3}\left(
x^{3}\right)  \right)  \right]  =T^{\left(  k+1\right)  }\left(  x^{1}%
,x^{3}\right)  .
\]
It is also easy to check that:%
\[
\forall m\in\mathbb{N}_{0}:T^{\left(  m\right)  }\left(  x^{1},x^{3}\right)
=m\Rightarrow\left(  \forall i>m:T^{\left(  i\right)  }\left(  x^{1}%
,x^{3}\right)  =m\right)  .
\]
Hence the induction has been completed.

Given (\ref{eq03023}), we see immediately that
\[
\forall\ \left(  x^{1},x^{3}\right)  \in V\times V,i\geq\left\vert
V\right\vert :T^{\left(  i\right)  }\left(  x^{1},x^{3}\right)  =T^{\ast
}\left(  x^{1},x^{3}\right)
\]
which implies that both $\lim_{i\rightarrow\infty}T^{\left(  i\right)
}\left(  x^{1},x^{3}\right)  =T^{\ast}\left(  x^{1},x^{3}\right)  $ and
$\lim_{i\rightarrow\infty}\sigma^{\left(  i\right)  }\left(  x^{1}%
,x^{3}\right)  $ exist. Taking the limit (as $i$ tends to $\infty$)\ in
(\ref{eq03024})-(\ref{eq03025}) we get the \emph{optimality equations}%
\begin{align*}
T^{\ast}\left(  x^{1},x^{3}\right)   &  =\min_{x^{\prime}\in N\left[
x^{1}\right]  }\left[  1+T^{\ast}\left(  x^{\prime},\overline{\sigma}%
_{3}\left(  x^{3}\right)  \right)  \right] \\
\overline{\sigma}^{\ast}\left(  x^{1},x^{3}\right)   &  =\arg\min_{x^{\prime
}\in N\left[  x^{1}\right]  }\left[  1+T^{\ast}\left(  x^{\prime}%
,\overline{\sigma}_{3}\left(  x^{3}\right)  \right)  \right]
\end{align*}
hence, it is clear from the iteration (\ref{eq03024})-(\ref{eq03025}) that
$T_{\overline{\sigma}^{\ast}}\left(  x^{1},x^{3}\right)  =T^{\ast}\left(
x^{1},x^{3}\right)  $, for all$\ \left(  x^{1},x^{3}\right)  \in V\times V$.
\end{proof}

Now let us use $\overline{\sigma}^{\ast}$ of Lemma \ref{prop0303} to define
strategies $\overline{\sigma}_{i}^{\ast}$ for $C_{i}$ ($i\in\left\{
1,2\right\}  $) as follows:%
\begin{align*}
\forall\left(  x^{1},x^{2},x^{3}\right)   &  \in\mathbf{S}_{a}:\overline
{\sigma}_{1}^{\ast}\left(  x^{1},x^{2},x^{3}\right)  =\overline{\sigma}^{\ast
}\left(  x^{1},x^{3}\right)  ,\\
\forall\left(  x^{1},x^{2},x^{3}\right)   &  \in\mathbf{S}_{a}:\overline
{\sigma}_{2}^{\ast}\left(  x^{1},x^{2},x^{3}\right)  =\overline{\sigma}^{\ast
}\left(  x^{2},x^{3}\right)  .
\end{align*}
Then the following holds.

\begin{theorem}
\label{prop0304}Given some graph $G=\left(  V,E\right)  $, suppose SCPR is
played on $G$ and the robber is controlled by an oblivious deterministic
strategy $\overline{\sigma}_{3}$. Then $\overline{\sigma}_{i}^{\ast}$ is an
optimal strategy for $C_{i}$ ($i\in\left\{  1,2\right\}  $), for every
starting position $s=\left(  x^{1},x^{2},x^{3}\right)  \in\mathbf{S}_{a}$.
Furthermore
\[
\forall s=\left(  x^{1},x^{2},x^{3}\right)  \in\mathbf{S}_{a}:%
\begin{array}
[c]{lll}%
T_{\overline{\sigma}_{1}^{\ast}}\left(  x^{1},x^{3}\right)  \leq
T_{\overline{\sigma}_{2}^{\ast}}\left(  x^{2},x^{3}\right)  & \Rightarrow &
v\left(  s\right)  =1\\
T_{\overline{\sigma}_{1}^{\ast}}\left(  x^{1},x^{3}\right)  >T_{\overline
{\sigma}_{2}^{\ast}}\left(  x^{2},x^{3}\right)  & \Rightarrow & v\left(
s\right)  =0
\end{array}
.
\]

\end{theorem}

\begin{proof}
The key fact is this: when $\overline{\sigma}_{3}$ is oblivious deterministic,
the players $C_{1}$ and $C_{2}$ interact only at the last phase of the game,
when $R$ is captured. In effect each cop plays a \textquotedblleft
decoupled\textquotedblright\ classic CR\ game, in which $\overline{\sigma
}^{\ast}$ of Lemma \ref{prop0303} guarantees capture in minimum time. Of
course in the full SCPR game there is always the possibility that the other
cop can capture $R$ at an earlier time. Hence the best $C_{i}$ can do is to
attempt to capture $R$ at the earliest possible time and an optimal strategy
to this end is $\overline{\sigma}_{i}^{\ast}$; he has no incentive to deviate
from $\overline{\sigma}_{i}^{\ast}$ (by using another deterministic or
probabilistic strategy) because this can never reduce his projected capture
time. Hence $\overline{\sigma}_{i}^{\ast}$ is optimal for $C_{1}$. Since
$\overline{\sigma}_{1}^{\ast}$, $\overline{\sigma}_{2}^{\ast}$ and
$\overline{\sigma}_{3}$ are deterministic, the outcome of the game is also
deterministic. In particular, when $T_{\overline{\sigma}_{1}^{\ast}}\left(
x^{1},x^{3}\right)  \leq T_{\overline{\sigma}_{2}^{\ast}}\left(  x^{2}%
,x^{3}\right)  $, with probability 1 $C_{1}$ reaches $R$ before or at the same
time as $C_{2}$; hence $v\left(  s\right)  =1$; when $T_{\overline{\sigma}%
_{1}^{\ast}}\left(  x^{1},x^{3}\right)  >T_{\overline{\sigma}_{2}^{\ast}%
}\left(  x^{2},x^{3}\right)  $, $C_{2}$ reaches $R$ before $C_{1}$ with
probability 1; hence $v\left(  s\right)  =0$.
\end{proof}

The next theorem gives an additional characterization of the value $v\left(
s\right)  $. In the statement of the theorem we will use the following
notation:\ suppose the game is in the state $s$, $C_{1}$ plays $a^{1}$,
$C_{2}$ plays $a^{2}$ and $R$ plays the (predetermined)\ move $\overline
{\sigma}_{3}\left(  s\right)  $; then we denote the next game state by
$\widehat{\mathbf{T}}\left(  s,\left(  a^{1},a^{2},\overline{\sigma}%
_{3}\left(  s\right)  \right)  \right)  $. We have the following.

\begin{theorem}
\label{prop0305}Given some graph $G=\left(  V,E\right)  $, suppose SCPR is
played on $G$ and the robber is controlled by an oblivious deterministic
strategy $\overline{\sigma}_{3}$. Then, $\forall s\in\mathbf{S}_{a}$, we have%
\begin{equation}
v\left(  s\right)  =\max_{a^{1}}\min_{a^{2}}\left[  q\left(  s\right)
+v\left(  \widehat{\mathbf{T}}\left(  s,\left(  a^{1},a^{2},\overline{\sigma
}_{3}\left(  s\right)  \right)  \right)  \right)  \right]  =\min_{a^{2}}%
\max_{a^{1}}\left[  q\left(  s\right)  +v\left(  \widehat{\mathbf{T}}\left(
s,\left(  a^{1},a^{2},\overline{\sigma}_{3}\left(  s\right)  \right)  \right)
\right)  \right]  .\label{eq03031}%
\end{equation}

\end{theorem}

\begin{proof}
Since $\overline{\sigma}_{3}$ is deterministic, $\Pr\left(  \widehat
{\mathbf{T}}\left(  s,\left(  a^{1},a^{2},\overline{\sigma}_{3}\left(
s\right)  \right)  \right)  |s,a^{1},a^{2}\right)  =1$. Hence, by
\cite[Theorem 4.4.3]{filar1997competitive}:%
\[
v\left(  s\right)  =\mathbf{Val}\left[  q\left(  s\right)  +\sum_{s^{\prime
}\in S}\Pr\left(  s^{\prime}|s,a_{1},a_{2}\right)  v\left(  s^{\prime}\right)
\right]  =\mathbf{Val}\left[  q\left(  s\right)  +v\left(  \widehat
{\mathbf{T}}\left(  s,\left(  a^{1},a^{2},\overline{\sigma}_{3}\left(
s\right)  \right)  \right)  \right)  \right]  .
\]
Since $\overline{\sigma}_{1}^{\ast}$ and $\overline{\sigma}_{2}^{\ast}$ are
also deterministic, at every turn of the game they produce an action with
probability one. Hence there exist actions $\overline{a}^{1}=\overline{\sigma
}_{1}^{\ast}\left(  s\right)  $, $\overline{a}^{2}=\overline{\sigma}_{2}%
^{\ast}\left(  s\right)  $ such that%
\[
v\left(  s\right)  =q\left(  s\right)  +v\left(  \widehat{\mathbf{T}}\left(
s,\left(  \overline{a}^{1},\overline{a}^{2},\overline{\sigma}_{3}\left(
s\right)  \right)  \right)  \right)  .
\]
Since, from Theorem \ref{prop0304}, $v\left(  s\right)  \in\left\{
0,1\right\}  $, we consider two cases.

\begin{enumerate}
\item Suppose $v\left(  s\right)  =1$. This means, that starting at $s$,
$C_{1}$ will certainly capture $R$.

\begin{enumerate}
\item If $s$ is a $C_{1}$-capture state, then $q\left(  s\right)  =1$ and, for
any actions $\overline{a}^{1},\overline{a}^{2}$, $\widehat{\mathbf{T}}\left(
s,\left(  \overline{a}^{1},\overline{a}^{2},\overline{\sigma}_{3}\left(
s\right)  \right)  \right)  =\tau$, in which case
\[
v\left(  \widehat{\mathbf{T}}\left(  s,\left(  \overline{a}^{1},\overline
{a}^{2},\overline{\sigma}_{3}\left(  s\right)  \right)  \right)  \right)
=v\left(  \tau\right)  =0.
\]
Hence $v\left(  s\right)  $= $\max_{a^{1}}\min_{a^{2}}\left[  q\left(
s\right)  +v\left(  \widehat{\mathbf{T}}\left(  s,\left(  a^{1},a^{2}%
,\overline{\sigma}_{3}\left(  s\right)  \right)  \right)  \right)  \right]
=1$.

\item If $s$ is not a $C_{1}$-capture state, then $q\left(  s\right)  =0$ and
$v\left(  \widehat{\mathbf{T}}\left(  s,\left(  \overline{a}^{1},\overline
{a}^{2},\overline{\sigma}_{3}\left(  s\right)  \right)  \right)  \right)  =1$.
Suppose there existed some $\widehat{a}^{2}$ such that $v\left(
\widehat{\mathbf{T}}\left(  s,\left(  \overline{a}^{1},\widehat{a}%
^{2},\overline{\sigma}_{3}\left(  s\right)  \right)  \right)  \right)  =0$.
This would mean that, starting at $\widehat{\mathbf{T}}\left(  s,\left(
\overline{a}^{1},\widehat{a}^{2},\overline{\sigma}_{3}\left(  s\right)
\right)  \right)  $, $C_{2}$ would certainly capture $R$ before $C_{1}$ and,
since $\overline{a}^{1}$ is the optimal (fastest capturing) move for $C_{1}$,
we would also have
\[
\forall a^{1}\in\mathbf{A}_{1}\left(  s\right)  :q\left(  s\right)  +v\left(
\widehat{\mathbf{T}}\left(  s,\left(  a^{1},\widehat{a}^{2},\overline{\sigma
}_{3}\left(  s\right)  \right)  \right)  \right)  =0.
\]
But then $v\left(  s\right)  =\mathbf{Val}\left[  q\left(  s\right)  +v\left(
\widehat{\mathbf{T}}\left(  s,\left(  a^{1},a^{2},\overline{\sigma}_{3}\left(
s\right)  \right)  \right)  \right)  \right]  =0$, contrary to the assumption.
So we must instead have
\[
\forall a^{2}\in\mathbf{A}_{2}\left(  s\right)  :q\left(  s\right)  +v\left(
\widehat{\mathbf{T}}\left(  s,\left(  \overline{a}^{1},a^{2},\overline{\sigma
}_{3}\left(  s\right)  \right)  \right)  \right)  =1
\]
which implies $v\left(  s\right)  $= $\max_{a^{1}}\min_{a^{2}}\left[  q\left(
s\right)  +v\left(  \widehat{\mathbf{T}}\left(  s,\left(  a^{1},a^{2}%
,\overline{\sigma}_{3}\left(  s\right)  \right)  \right)  \right)  \right]
=1$.
\end{enumerate}

\item Now suppose $v\left(  s\right)  =0$. Then $s$ is not a $C_{1}$-capture
state, i.e., $q\left(  s\right)  =0$. Now, we will show that
\begin{equation}
\forall a^{1}\in\mathbf{A}_{1}\left(  s\right)  :\exists a^{2}\in
\mathbf{A}_{2}\left(  s\right)  :v\left(  \widehat{\mathbf{T}}\left(
s,\left(  a^{1},a^{2},\overline{\sigma}_{3}\left(  s\right)  \right)  \right)
\right)  =0. \label{eq03026}%
\end{equation}
If this is not the case, then we must have
\[
\exists\widetilde{a}^{1}\in\mathbf{A}_{1}\left(  s\right)  :\forall a^{2}%
\in\mathbf{A}_{2}\left(  s\right)  :v\left(  \widehat{\mathbf{T}}\left(
s,\left(  \widetilde{a}^{1},a^{2},\overline{\sigma}_{3}\left(  s\right)
\right)  \right)  \right)  =1.
\]
Then $C_{1}$ will certainly capture $R$ (before $C_{2}$) starting from the
game position $\mathbf{T}\left(  s,\left(  \widetilde{a}^{1},a^{2}%
,\overline{\sigma}_{3}\left(  s\right)  \right)  \right)  $ and this will be
true for any $a^{2}\in\mathbf{A}_{2}\left(  s\right)  $. But this means that
$C_{1}$, starting from game position $s$ and playing $\widetilde{a}^{1}$, will
certainly capture $R$ before $C_{2}$; which in turn means $v\left(  s\right)
=1$, contrary to the hypothesis. Hence (\ref{eq03026})\ holds and this implies%
\begin{align*}
\forall a^{1}  &  \in\mathbf{A}_{1}\left(  s\right)  :\min_{a^{2}}v\left(
\widehat{\mathbf{T}}\left(  s,\left(  a^{1},a^{2},\overline{\sigma}_{3}\left(
s\right)  \right)  \right)  \right)  =0\\
&  \Rightarrow\max_{a^{1}}\min_{a^{2}}v\left(  \widehat{\mathbf{T}}\left(
s,\left(  a^{1},a^{2},\overline{\sigma}_{3}\left(  s\right)  \right)  \right)
\right)  =0.
\end{align*}

\end{enumerate}

\noindent Hence we have proved the first part of (\ref{eq03031}). The proof of
the second part is similar and omitted.
\end{proof}

\begin{remark}
\normalfont It must be emphasized that Theorem \ref{prop0304} and Theorem
\ref{prop0305} do \emph{not} hold for deterministic \emph{non}-oblivious
strategies $\overline{\sigma}_{3}$. This can be seen by the following
counterexample. Suppose that concurrent SCPR\ is played on the graph of Figure
1, starting from the state $\left(  2,6,1\right)  $.

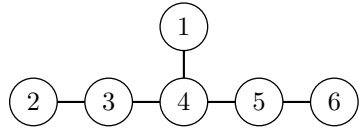
\begin{figure}[H]
\begin{center}
\begin{tikzpicture}
\SetGraphUnit{2}
\Vertex[x= 0,y=1]{1}
\Vertex[x=-2,y=0]{2}
\Vertex[x=-1,y=0]{3}
\Vertex[x= 0,y=0]{4}
\Vertex[x= 1,y=0]{5}
\Vertex[x= 2,y=0]{6}
\Edge(1)(4)
\Edge(2)(3)
\Edge(3)(4)
\Edge(4)(5)
\Edge(5)(6)
\SetVertexNoLabel
\end{tikzpicture}
\end{center}
\par
\label{fig01}\caption{An example where deterministic robber strategy results
in randomized optimal cop strategies.}%
\end{figure}\noindent Furthermore, the robber is controlled by the
$\overline{\sigma}_{3}$ which is (partially)\ described in the following table.

\begin{table}[H]
\begin{center}%
\begin{tabular}
[c]{|c|c|}\hline
$\left(  x_{t}^{1},x_{t}^{2},x_{t}^{3}\right)  $ & $x_{t+1}^{3}=\overline
{\sigma}_{3}\left(  x_{t}^{1},x_{t}^{2},x_{t}^{3}\right)  $\\\hline
$\left(  2,6,1\right)  $ & 4\\\hline
$\left(  2,6,4\right)  $ & 3\\\hline
$\left(  2,5,4\right)  $ & 5\\\hline
$\left(  3,6,4\right)  $ & 5\\\hline
$\left(  3,5,4\right)  $ & 3\\\hline
\end{tabular}
\end{center}
\caption{A part of the robber strategy $\sigma_{3}$}%
\end{table}

\noindent For every game state not listed above the robber stays in place,
i.e., $x_{t+1}^{3}=\overline{\sigma}_{3}\left(  x_{t}^{1},x_{t}^{2},x_{t}%
^{3}\right)  =x_{t}^{3}$. Now consider what the first moves of $C_{1}$ and
$C_{2}$ should be. They know that $R$ will move into vertex 4; $C_{1}$ can
either stay at 2 or move into 3; $C_{2}$ can either stay at 6 or move into 5.
After the first move is completed, the possible game states are the following.

\begin{table}[H]
\begin{center}%
\begin{tabular}
[c]{|c|c|c|c|c|}\hline
$s_{0}=\left(  2,6,1\right)  $ & $a_{1}^{1}=2$ & $a_{1}^{2}=6$ & $a_{1}%
^{3}=\overline{\sigma}_{3}\left(  2,6,1\right)  =4$ & $s_{1}=\left(
2,6,4\right)  $\\\hline
$s_{0}=\left(  2,6,1\right)  $ & $a_{1}^{1}=2$ & $a_{1}^{2}=5$ & $a_{1}%
^{3}=\overline{\sigma}_{3}\left(  2,6,1\right)  =4$ & $s_{1}=\left(
2,5,4\right)  $\\\hline
$s_{0}=\left(  2,6,1\right)  $ & $a_{1}^{1}=3$ & $a_{1}^{2}=6$ & $a_{1}%
^{3}=\overline{\sigma}_{3}\left(  2,6,1\right)  =4$ & $s_{1}=\left(
3,6,4\right)  $\\\hline
$s_{0}=\left(  2,6,1\right)  $ & $a_{1}^{1}=3$ & $a_{1}^{2}=5$ & $a_{1}%
^{3}=\overline{\sigma}_{3}\left(  2,6,1\right)  =4$ & $s_{1}=\left(
3,5,4\right)  $\\\hline
\end{tabular}
\end{center}
\caption{Possible states at the end of the first turn.}%
\end{table}

\noindent It is easy to check (from the respective $\overline{\sigma}_{3}$
values)\ that for $s_{1}=\left(  2,6,4\right)  $ and $s_{1}=\left(
3,5,4\right)  $ the capturing cop is $C_{1}$, while for $s_{1}=\left(
2,5,4\right)  $ and $s_{1}=\left(  3,6,4\right)  $ the capturing cop is
$C_{2}$. Hence the game can be written out as the following (one-turn)\ matrix game

\begin{table}[H]
\begin{center}%
\begin{tabular}
[c]{|c|c|c|}\hline
& $a^{2}=6$ & $a^{2}=5$\\\hline
$a^{1}=2$ & 1 & 0\\\hline
$a^{1}=3$ & 0 & 1\\\hline
\end{tabular}
\end{center}
\caption{The one-turn matrix game equivalent to the original stochastoc game.}%
\end{table}

\noindent It is easy to compute, using standard methods, that the optimal
strategies for this game. $C_{1}$ must use $\Pr\left(  a^{1}=2\right)
=\Pr\left(  a^{1}=3\right)  =\frac{1}{2}$ and $C_{2}$ must use $\Pr\left(
a^{2}=6\right)  =\Pr\left(  a^{2}=5\right)  =\frac{1}{2}$. This implies that
the optimal strategies $\overline{\sigma}_{1}^{\ast}$ and $\overline{\sigma
}_{2}^{\ast}$ are randomized, despite the fact that $\overline{\sigma}_{3}$ is
deterministic (but not oblivious). Many similar examples can be constructed.
The important point is this:\ when $\overline{\sigma}_{3}$ is not oblivious,
$C_{1}$ (resp. $C_{2}$)\ moves can influence future $R$ moves and (since moves
are performed simultaneously) this influence cannot be predicted by $C_{2}$
(resp. $C_{1}$).
\end{remark}

\section{Related Work\label{sec05}}

In this section we present work which is related to both the SCPR and other
variants of the CR game.

We have already mentioned that the interested reader can find useful
references to the CR\ literature in the book \cite{nowakowskibonato} by
Nowakowski and Bonato. The CR\ literature is mainly oriented to graph
theoretic and combinatorial considerations. Indeed CR\ can be seen as a
\emph{combinatorial game}. On the topic of combinatorial games, the reader can
consult the introductory text \cite{albert} as well as the classic book (in
four volumes) \cite{berlekamp} by Berlekamp and Conway. We also find
interesting combinatorial generalizations of the CR\ game in the papers
\cite{bonato01,bonato02} by A. Bonato and G. MacGillivray.

We believe that \textquotedblleft classic\textquotedblright\ game theory
offers a natural (but not often used in the \textquotedblleft
mainstream\textquotedblright\ CR literature)\ framework for the analysis of
CR\ games. In particular, as already seen, we consider SCPR as a
\emph{stochastic game}. Stochastic games were introduced by Shapley
\cite{shapley1953stochastic}. A classic book on the subject is
\cite{filar1997competitive}, which also contains a rich bibliography; see also
\cite{mertens1992stochastic}.

A type of stochastic games which are especially related to CR\ games are
\emph{recursive games}, in which whenever a non-zero-payoff is received the
play immediately moves to an absorbing state. Recursive games were introduced
by Everett \cite{everett}. It is obvious that SCPR is a recursive game; while
we have not used results from the recursive game theory in the current paper,
we believe this may turn out to be a fruitful connection.

Let us now mention a construction which has been used in several
\textquotedblleft classic\textquotedblright\ CR papers
\cite{bonato01,bonato02,hahn2006game}. Suppose that a \textquotedblleft
classic\textquotedblright\ CR\ game is played between one cop and one robber
on the undirected graph $G=\left(  V,E\right)  $. We now construct the
\emph{game digraph} $D=\left(  S,A\right)  $, where the vertex set is
$S=V\times V \times\{1,2\}$ and the arc set $A$ encodes possible
vertex-to-vertex transitions. Then a play of the CR\ game can be understood as
a walk on $D$; the cop wins if he can force the walk to pass through a vertex
of the form $\left(  x,x,i\right)  $. Hence CR can be seen as a game in which
the two players push a token along the arcs of the digraph. As pointed out in
\cite{bonato01,bonato02} many CR\ variants and several other pursuit games on
graphs (including the concurrent CR game) can be formulated in a similar manner.

It turns out that such \textquotedblleft digraph games\textquotedblright\ have
been studied by several researchers and the related literature is spread among
many communities. The earliest such works of which we are aware is
\cite{berarducci,macnaughton}. Other early examples of this iterature are the
papers \cite{baston1993,ehrenfeucht,washburn}. But probably the most
widespread application of this point of view is in the literature of
reachability games \cite{berwangergraph} and, more generally, $\omega
$\emph{-regular games} \cite{mazala2002infinite}. In a reachability game two
players take turns moving a token along the arcs of a digraph; player 1 wants
to place the token on one of the nodes of a subset of the digraph vertices
while player 2 wants to avoid this event. In addition to \textquotedblleft
classic\textquotedblright\ sequential rechability games, many other variants
have been studied, e.g., stochastic \cite{chatterjee}, concurrent
\cite{alfaro2007concurrent}, $n$-player \cite{chatterjeeNplayer} etc. The
connection to CR\ games is obvious; it seems likely that the voluminous
literature on reachability games contains results of interest to CR\ researchers.

\section{Conclusion\label{sec06}}

We have introduced the game of selfish cops and passive robber
(SCPR\ game)\ and established its basic properties, namely the existence of
value and optimal strategies for both the sequential and concurrent variants;
we have also provided algorithms for the computation of the aforementioned
quantities. In the current paper we have examined \emph{qualitative} variants
of the game, i.e., these in which the goal of the cops is simply to capture
the robber. In a forthcoming paper we will examine \emph{quantitative}
variants, in which the goal is to capture the robber \emph{in the shortest
possible time}.

Several additional issues merit further study and will be the subject of our
future research. We have formulated SCPR\ as a zero-sum game; but reasonable
formulations as a \emph{non-zero-sum} game are also possible and we conjecture
that these may lead to qualitatively different results. In addition, if we
remove the assumption that the robber is passive and deal instead with the
situation of two selfish robbers and a robber \emph{actively trying to avoid
capture}, we are left with a \emph{three}-player game, which we intend to
study in the future.


\begin{thebibliography}{99}                                                                                               %


\bibitem {albert}M.H. Albert, R, Nowakowski and D. Wolfe (2007). \emph{Lessons
in Play: An Introduction to Combinatorial Game Theory}.

\bibitem {alfaro2007concurrent}L. de Alfaro, T. A. Henzinger and O. Kupferman.
\textquotedblleft Concurrent reachability games\textquotedblright.
\emph{Theoretical Computer Science,} vol.386 (2007), pp. 188-217.

\bibitem {baston1993}V. J. Baston and F. A. Bostock. \textquotedblleft
Infinite deterministic graphical games\textquotedblright. \emph{SIAM journal
on control and optimization} vol.31 (1993). pp.1623-1629.

\bibitem {berarducci}A. Berarducci and B. Intrigila. \textquotedblleft On the
cop number of a graph\textquotedblright. \emph{Advances in Applied
Mathematics}, vol. 14 (1993), pp. 389-403.

\bibitem {berlekamp}E. Berlekamp, J. H. Conway and R. Guy (1982).
\emph{Winning Ways for your Mathematical Plays}.

\bibitem {berwangergraph}D. Berwanger, Graph games with perfect information,
\emph{preprint}.

\bibitem {bonato01}A. Bonato and G. MacGillivray. \textquotedblleft A general
framework for discrete - time pursuit games\textquotedblright, \emph{preprint}.

\bibitem {bonato02}A. Bonato and G. MacGillivray. \textquotedblleft
Characterizations and algorithms for generalized Cops and Robbers
games\textquotedblright, accepted to \emph{Contributions to Discrete
Mathematics} (2016).

\bibitem {chatterjee}K. Chatterjee and T.A. Henzinger. \textquotedblleft A
survey of stochastic $\omega$-regular games\textquotedblright. \emph{Journal
of Computer and System Sciences}, vol.78 (2012), pp. 394-413.

\bibitem {chatterjeeNplayer}K. Chatterjee, R. and M. Jurdzi\'{n}ski.
\textquotedblleft On Nash equilibria in stochastic games.\textquotedblright%
\ \emph{International Workshop on Computer Science Logic}. Springer Berlin
Heidelberg, 2004.

\bibitem {ehrenfeucht}A. Ehrenfeucht and J. Mycielski. \textquotedblleft
Positional strategies for mean payoff games\textquotedblright.
\emph{International Journal of Game Theory}, vol.8 (1979), pp. 109-113.

\bibitem {everett}H. Everett. \textquotedblleft Recursive
games\textquotedblright. \emph{Contributions to the Theory of Games}, vol.3
(1957), pp. 47-78.

\bibitem {filar1997competitive}J. Filar K. Vrieze. \emph{Competitive Markov
decision processes}. Springer Science \& Business Media, 1997.

\bibitem {hahn2006game}G. Hahn and G. MacGillivray, \textquotedblleft A note
on $k$-cop, $l$-robber games on graphs\textquotedblright. \emph{Discrete
Mathematics}, vol.306 (2006), pp.2492--2497.

\bibitem {kehagias2016}Ath. Kehagias and G. Konstantinidis. \textquotedblleft
Simultaneously moving cops and robbers\textquotedblright. \emph{Theoretical
Computer Science}, vol. 645 (2016), pp.48-59.

\bibitem {macnaughton}R. McNaughton. \textquotedblleft Infinite games played
on finite graphs\textquotedblright. \emph{Annals of Pure and Applied
Logic}\textbf{,} vol. 65 (1993), pp. 149-184.

\bibitem {mazala2002infinite}R. Mazala, \textquotedblleft Infinite games." In
\emph{Automata logics, and infinite games} (2002), pp. 23-38.

\bibitem {mertens1992stochastic}J.-F. Mertens. \textquotedblleft Stochastic
games\textquotedblright. \emph{Handbook of game theory with economic
applications}, vol.3 (2002), pp. 1809-1832.

\bibitem {nowakowski1983vertex}R. Nowakowski and P. Winkler. \textquotedblleft
Vertex to vertex pursuit in a graph\textquotedblright. \emph{Discrete
Mathematics}, vol. 43 (1983), pp. 230--239.

\bibitem {nowakowskibonato}R. Nowakowski and A. Bonato, \emph{The Game of Cops
and Robbers on Graphs}, AMS, 2011.

\bibitem {quilliot1978jeux}A. Quilliot, \emph{Jeux et pointes fixes sur les
graphes}, Ph.D. Dissertation, Universite de Paris VI, 1978.

\bibitem {osborne1994game}M.J. Osborne and A. Rubinstein. \emph{A Course in
Game Theory}. MIT Press, 1994.

\bibitem {shapley1953stochastic}L.S. Shapley. \textquotedblleft Stochastic
games\textquotedblright. \emph{Proceedings of the National Academy of Sciences
of the United States of America,} vol.39 (1953), pp. 1095-1100.

\bibitem {washburn}A. Washburn. \textquotedblleft Deterministic graphical
games\textquotedblright. \emph{Journal of Mathematical Analysis and
Applications}\textbf{,} vol.153 (1990), pp. 84-96.
\end{thebibliography}
\end{document}